\documentclass[a4paper,UKenglish,cleveref, autoref, thm-restate]{lipics-v2021}
\nolinenumbers

\usepackage{enumerate}
\usepackage{amsmath,amssymb,amsfonts,latexsym}
\usepackage{stmaryrd}
\usepackage{cleveref}

\newcommand{\emptystr}{\varepsilon}

\newcommand{\cgen}{c}
\newcommand{\cint}{\hat{c}}
\newcommand{\expand}[1]{\langle#1\rangle}
\newcommand{\numtr}{m_{\mathsf{tr}}}

\newcommand{\litf}[2]{\ensuremath{\mathit{f}_{#1, #2}}}
\newcommand{\litp}[1]{\ensuremath{\mathit{p}_{#1}}}
\newcommand{\litrefa}[3]{\ensuremath{\mathit{ref}^{\mathrm{A}}_{#1 \leftarrow #2, #3}}}
\newcommand{\litrefb}[3]{\ensuremath{\mathit{ref}^{\mathrm{B}}_{#1 \leftarrow #2, #3}}}
\newcommand{\litrefc}[4]{\ensuremath{\mathit{ref}^{\mathrm{C}}_{#1, #2 \leftarrow #3, #4}}}
\newcommand{\litdref}[3]{\ensuremath{\mathit{dref}_{#1, #2\leftarrow#3}}}
\newcommand{\litdepth}[3]{\ensuremath{\mathit{depth}_{#1, #2, #3}}}
\newcommand{\litq}[2]{\ensuremath{\mathit{q}_{#1, #2}}}
\newcommand{\sub}[1]{\llbracket #1 \rrbracket}

\newcommand{\hbrk}[1]{[#1)}

\newcommand{\idtt}[1]{\ensuremath{\mathtt{#1}}}

\title{On the Smallest Size of Internal Collage Systems}

\author{Soichiro Migita}{Kyushu Institute of Technology, Japan}{migita.soichiro595@mail.kyutech.jp}{}{}
\author{Kyotaro Uehata}{Kyushu Institute of Technology, Japan}{}{}{}
\author{Tomohiro I}{Kyushu Institute of Technology, Japan}{tomohiro@ai.kyutech.ac.jp}{https://orcid.org/0000-0001-9106-6192}{KAKENHI 24K02899, JST AIP Acceleration Research JPMJCR24U4}

\titlerunning{On the Smallest Size of Internal Collage Systems}

\authorrunning{S. Migita, K. Uehata and T. I}

\Copyright{Soichiro Migita, Kyotaro Uehata and Tomohiro I}



\begin{document}
\maketitle              % typeset the header of the contribution
\begin{abstract}
A Straight-Line Program (SLP) for a string $T$ is a context-free grammar in Chomsky normal form that derives $T$ only,
which can be seen as a compressed form of $T$.
Kida et al.\ introduced collage systems [Theor. Comput. Sci., 2003]
to generalize SLPs by adding repetition rules and truncation rules.
The smallest size $c(T)$ of collage systems for $T$ has gained attention
to see how these generalized rules improve the compression ability of SLPs.
Navarro et al. [IEEE Trans. Inf. Theory, 2021] showed that $c(T) \in O(z(T))$ and 
there is a string family with $c(T) \in \Omega(b(T) \log |T|)$, 
where $z(T)$ is the number of phrases in the Lempel-Ziv parsing of $T$
and $b(T)$ is the smallest size of bidirectional schemes for $T$.
They also introduced a subclass of collage systems, called internal collage systems,
and proved that its smallest size $\hat{c}(T)$ for $T$ is at least $b(T)$.
While $c(T) \le \hat{c}(T)$ is obvious, it is unknown how large $\hat{c}(T)$ is compared to $c(T)$.
In this paper, we prove that $\hat{c}(T) = \Theta(c(T))$ by showing that
any collage system of size $m$ can be transformed into an internal collage system of size $O(m)$ in $O(m^2)$ time.
Thanks to this result, we can focus on internal collage systems to study the asymptotic behavior of $c(T)$,
which helps to suppress excess use of truncation rules.
As a direct application, we get $b(T) = O(c(T))$, which answers an open question posed in [Navarro et al., IEEE Trans. Inf. Theory, 2021].
We also give a MAX-SAT formulation to compute $\hat{c}(T)$ for a given $T$.
\end{abstract}

\section{Introduction}\label{sec:intro}

Grammar-based compression is a framework of lossless compression whose outcomes are modeled by grammars like Context-free grammars (CFGs).
Grammar-based compression is popular because it is not only powerful to model the practical compressors such as LZ78~\cite{Ziv1978Coi} and RePair~\cite{1999LarssonM_OfflinDictionBasedCompr_DCC},
but also suitable to design algorithms working directly on compressed data~\cite{2012Lohrey_AlgorOnSlpComprStrin}.
In this paper, we focus on CFGs in Chomsky normal form that derive a single string, which are called \emph{Straight-Line Programs (SLPs)}~\cite{1995KarpinskiRS_PatterMatchForStrinWith_CPM}.
\footnote{The definition of SLPs varies across the literature. For example, SLPs are not necessarily in Chomsky normal form in~\cite{2012Lohrey_AlgorOnSlpComprStrin}.}

An SLP of size $m$ for a string $T$ has $m$ nonterminals
and the starting nonterminal is deterministically expanded to a unique string $T$ 
according to production rules of the form $X \rightarrow a$ or $X \rightarrow YZ$, where $X, Y, Z$ are nonterminals and $a$ is a terminal symbol.
The derived string from $X$, denoted by $\expand{X}$, is $\expand{X} = a$ for the former rules
and $\expand{X} = \expand{Y}\expand{Z}$ for the latter rules called the \emph{concatenation rules}.

Kida et al.~\cite{2003KidaMSTSA_CollagSystemUnifyFramewFor} introduced \emph{collage systems}
to generalize SLPs by adding repetition rules and truncation rules.
For a \emph{repetition rule} (a.k.a.\ run-length rule) $X \rightarrow Y^{r}$ with $r > 2$,
$X_i$ is expanded to the $r$-times repeat of $Y$ and we have $\expand{X} = \expand{Y}^r$.
For a \emph{truncation rule} $X \rightarrow Y[b..e)$ with $1 \le b < e \le |\expand{Y}| + 1$,
$\expand{X}$ is obtained by truncating $\expand{Y}$ to the substring $\expand{Y}[b..e)$, 
the substring of $\expand{Y}$ starting at $b$ and ending at $e-1$.
\footnote{In the original paper~\cite{2003KidaMSTSA_CollagSystemUnifyFramewFor}, the substring truncation is implemented by the prefix truncation and the suffix truncation.}
Collage systems with no truncation rules are also called \emph{Run-length SLPs (RLSLPs)}~\cite{2016NishimotoIIBT_FullyDynamDataStrucFor_MFCS}
and collage systems with no repetition rules are also called \emph{composition systems}~\cite{1996GasieniecKPR_EfficAlgorForLempelZip_SWAT}.
The size of a collage system is measured by the number of its nonterminals and 
let $\cgen(T)$ denote the smallest size of collage systems for a string $T$.

Recently, compressibility measures for highly-repetitive strings have been extensively studied~\cite{2021Navarro_IndexHighlRepetStrinCollec_I,2021Navarro_IndexHighlRepetStrinCollec_II}.
Although computing $\cgen(T)$ for a given string $T$ is NP-hard~\cite{2024KawamotoIKB_HardnOfSmallRlslpAnd_DCC},
it has gained attention as a compressibility measure 
to see how repetition and/or truncation rules improve the compression ability of SLPs.
Navarro et al.~\cite{2021NavarroOP_ApproxRatioOfOrderParsin} showed that $\cgen(T) \in O(z(T))$ and 
there is a string family with $\cgen(T) \in \Omega(b(T) \log |T|)$, 
where $z(T)$ is the size of LZ76 parsing~\cite{1976LempelZ_ComplOfFinitSequen_TIT} of $T$
and $b(T)$ is the smallest size of bidirectional schemes~\cite{1982StorerS_DataComprViaTexturSubst} for $T$.
They also introduced a subclass of collage systems, called \emph{internal collage systems}:
An internal collage system must have production rules so that, for every nonterminal $X$, 
the starting nonterminal can be expanded to get a sequence that contains $X$ without using truncation rules.
This restriction helps to suppress excess use of truncation rules.
In particular, a general collage system may have a nonterminal $X$ such that $\expand{X}$ is not a substring of $T$ but the substring of $\expand{X}$ is truncated and used to represent a part of $T$.
This is not the case in internal collage systems where every nonterminal $X$ can be reached from the starting nonterminal without truncation rules, implying that $\expand{X}$ appears in $T$ at least once.
Moreover, the restriction enables to define a parsing of $T$ of size linear to the size of an internal collage system
and prove its smallest size $\cint(T)$ is in $\Omega(b(T))$~\cite{2021NavarroOP_ApproxRatioOfOrderParsin}.

While $\cgen(T) \le \cint(T)$ is obvious, it is unknown how large $\cint(T)$ is compared to $\cgen(T)$.
In this paper, we prove that $\cint(T) = \Theta(\cgen(T))$ by showing that
any collage system of size $m$ can be transformed into an internal collage system of size $O(m)$ in $O(m^2)$ time.
Thanks to this result, we can focus on internal collage systems to study the asymptotic behavior of $\cgen(T)$.
As a direct application, we get $b(T) = O(\cgen(T))$, which answers an open question posed in~\cite{2021NavarroOP_ApproxRatioOfOrderParsin}.

Kawamoto et al.~\cite{2024KawamotoIKB_HardnOfSmallRlslpAnd_DCC} obtained a hardness result for computing $\cgen(T)$.
We just remark that exactly the same proof works to get the following result for $\cint(T)$ because truncation rules are not used to solve the reduced problem, and therefore, there are no difference between $\cgen(T)$ and $\cint(T)$ in the reduction.
\begin{theorem}
  The problem of computing $\cint(T)$ for a given string $T$ is NP-hard.
\end{theorem}
In this paper, we give a MAX-SAT formulation to compute $\cint(T)$ for a given $T$,
extending the previous approaches for SLPs~\cite{2022BannaiGIKKN_ComputNpHardRepetMeasur_ESA} and RLSLPs~\cite{2024KawamotoIKB_HardnOfSmallRlslpAnd_DCC}.
Since we now know that $\cint(T) = \Theta(\cgen(T))$, this provides a tool to study the behavior of both $\cint(T)$ and $\cgen(T)$ through computing exact values of $\cint(T)$ for concrete strings (although we may be able to process only short strings in a realistic time).

\section{Preliminaries}\label{sec:prelim}

\subsection{Basic notation}
An integer interval $\{ i, i+1, \dots, j\}$ is denoted by $[i..j]$, 
where $[i..j]$ represents the empty interval if $i > j$.
Also, $[i..j)$ denotes $[i..j-1]$.

Let $\Sigma$ be an ordered finite \emph{alphabet}.
An element of $\Sigma^*$ is called a \emph{string} over $\Sigma$.
The length of a string $w$ is denoted by $|w|$. 
The empty string $\emptystr$ is the string of length 0,
that is, $|\emptystr| = 0$.
Let $\Sigma^+ = \Sigma^* - \{\emptystr\}$ and $\Sigma^k = \{ w \in \Sigma^* \mid |w| = k \}$ for any non-negative integer $k$.
The concatenation of two strings $x$ and $y$ is denoted by $x \cdot y$ or simply $xy$.
When a string $w$ is represented by the concatenation of strings $x$, $y$ and $z$ (i.e., $w = xyz$), 
then $x$, $y$ and $z$ are called a \emph{prefix}, \emph{substring}, and \emph{suffix} of $w$, respectively.
A substring $x$ of $w$ is called \emph{proper} if $x \neq w$.
For any string $w$ and non-negative integer $k$, let $w^k$ denote the $k$-times repeat of $w$, i.e., 
$w^0 = \emptystr$ and $w^k = w^{k-1} \cdot w$ for any $k > 1$.

A \emph{factorization} of a string $w$ is a sequence $w_1, w_2, \dots, w_h$ of substrings of $w$ such that $w = w_1 w_2 \cdots w_h$.
Each substring $w_i~(1 \le i \le h)$ is called a \emph{factor} of the factorization.

The $i$-th symbol of a string $w$ is denoted by $w[i]$ for $1 \leq i \leq |w|$,
and the substring of a string $w$ that begins at position $i$ and
ends at position $j$ is denoted by $w[i..j]$ for $1 \leq i \leq j \leq |w|$,
i.e., $w[i..j] = w[i]w[i+1] \cdots w[j]$.
For convenience, let $w[i..j] = \emptystr$ if $j < i$.

\subsection{Collage Systems and Internal Collage Systems}
A collage system~\cite{2003KidaMSTSA_CollagSystemUnifyFramewFor} of size $m$ has $m$ nonterminals.
Each nonterminal $X$ derives a single string $\expand{X}$ that is defined by a unique rule to expand $X$ of either one of the following forms, where $Y$ and $Z$ are nonterminals:
\begin{itemize}
\item $X \rightarrow a$ is an \emph{atomic rule} such that $\expand{X} = a$, where $a \in \Sigma$.
\item $X \rightarrow YZ$ is a \emph{concatenation rule} such that $\expand{X} = \expand{Y}\expand{Z}$.
\item $X \rightarrow Y^{r}$ with $r > 2$ is a \emph{repetition rule} such that $\expand{X} = \expand{Y}^{r}$.
\item $X \rightarrow Y[b..e)$ with $1 \le b < e \le |\expand{Y}| + 1$ is a \emph{truncation rule} such that $\expand{X} = \expand{Y}[b..e)$.
\end{itemize}
We say that $X$ or its rule \emph{refers to} symbols in the righthand side of the rule.
We assume that nonterminals are sorted such that every nonterminal refers to smaller nonterminals, which is always possible because a collage system must not have a reference loop to represent a finite single string.
Also, we assume that there is no useless nonterminal, i.e., all nonterminals are involved in the process of expanding the starting nonterminal.
A collage system is called \emph{internal} if, for every nonterminal $X$, the starting nonterminal can be expanded to get a sequence that contains $X$ without using truncation rules.
Let $\cgen(T)$ (resp.\ $\cint(T)$) denote the smallest size of collage systems (resp.\ internal collage systems) for a string $T$.

We define the \emph{binary parse tree} of a collage system in a top-down manner as follows:
\begin{itemize}
\item The root node is labeled with the starting nonterminal.
\item If a node is labeled with $X$ that has atomic rule $X \rightarrow a$, its only child is labeled with $a$.
\item If a node is labeled with $X$ that has concatenation rule $X \rightarrow YZ$, 
  its left child is labeled with $Y$ and its right child is labeled with $Z$.
\item If a node is labeled with $X$ that has repetition rule $X \rightarrow Y^{r}$, 
  its left child is labeled with $Y$ and its right child is labeled with $Y^{r-1}$.
\item If a node is labeled with $X$ that has truncation rule $X \rightarrow Y[b..e)$, 
  its only child is labeled with $Y[b..e)$.
\item If a node does not fit in any of the cases above, it is a leaf.
\end{itemize}
Note that we treat $Y^{r-1}$, $Y[b..e)$ and $Y$ as different node labels,
and a node with label $Y^{r-1}$ or $Y[b..e)$ falls into the last case and always becomes a leaf in the binary parse tree.
A collage system is internal if and only if every nonterminal appears as a node label in the binary parse tree.

See \cref{fig:bpt} (resp. \cref{fig:noninternal_bpt}) for examples of an internal collage system (resp. non-internal collage system) and their binary parse tree.
\begin{figure}[t]
  \includegraphics[scale=0.46]{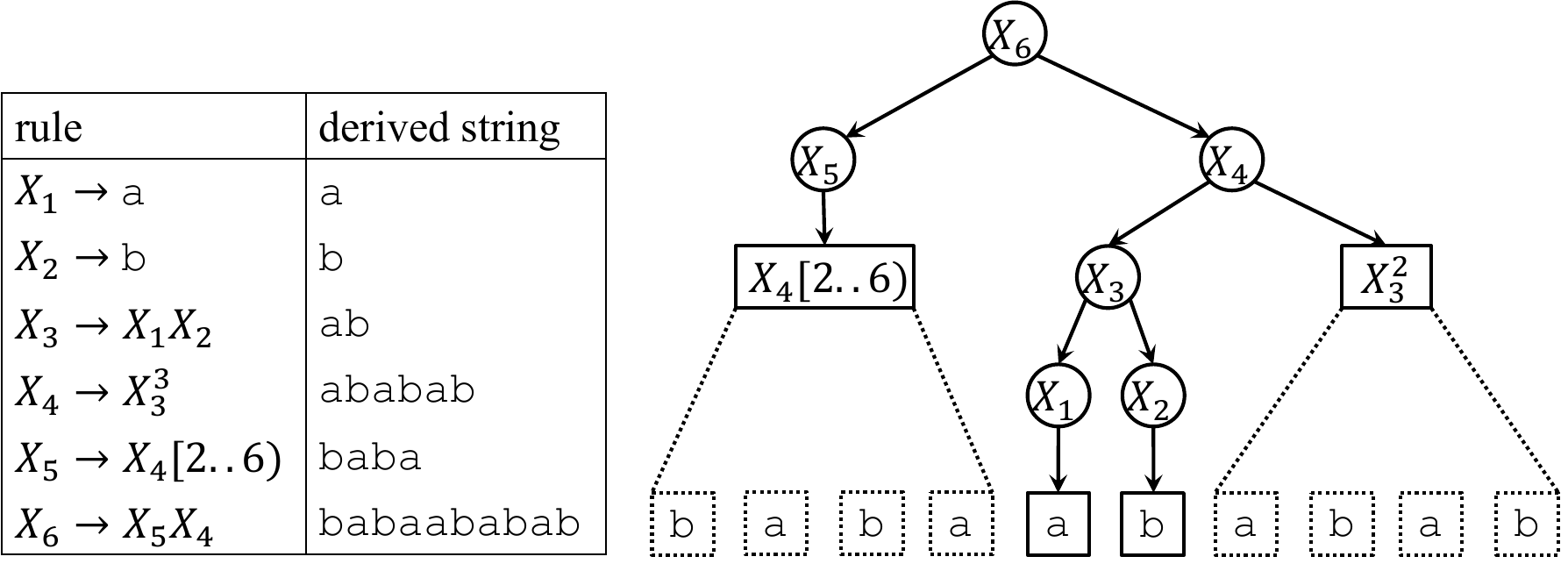}
  \caption{
    An illustration of the binary parse tree (right) for the internal collage system having six rules shown left.
    The internal nodes are depicted by circles, and the leaves by solid boxes.
    The characters derived from leaves are depicted with dotted boxes.
    This is an internal collage system because every nonterminal appears as a node label in the binary parse tree.
  }
  \label{fig:bpt}
\end{figure}

\begin{figure}[t]
  \includegraphics[scale=0.46]{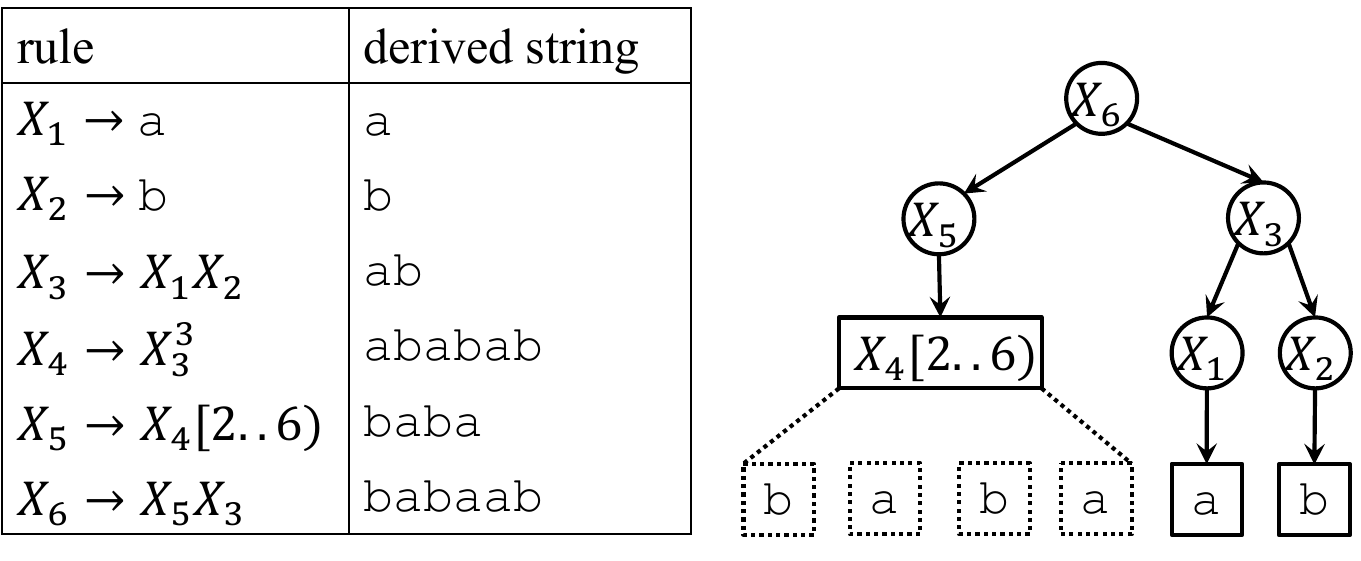}
  \caption{
    An illustration of the binary parse tree (right) for the non-internal collage system having six rules shown left.
    It is not an internal collage system because $X_4$ does not appear as a node label in the binary parse tree.
  }
  \label{fig:noninternal_bpt}
\end{figure}

Similarly to~\cite{2021NavarroOP_ApproxRatioOfOrderParsin}, we define the \emph{grammar tree} of a collage system as the tree obtained from the binary parse tree by deleting all descendant nodes under any node that is not the leftmost node with the same label (see \cref{fig:grammar_tree} for an example).\footnote{The idea of grammar trees for internal collage systems is extended from partial parse trees for SLPs~\cite{2003Rytter_ApplicOfLempelZivFactor}.}
Since the grammar tree has every nonterminal at most once as the label of an internal node, we sometimes identify an internal node by its label.
\begin{proposition}\label{prop:ics}
  The following statements hold for the grammar tree of an internal collage system of size $m$ that contains $\sigma$ atomic rules and $\numtr$ truncation rules:
  \begin{itemize}
     \item The number of internal nodes is $m$ because every nonterminal appears ``exactly once'' as the label of an internal node in the grammar tree.
     \item The number of leaves is $m - \numtr - \sigma + 1$ because every internal node for a concatenation or repetition rule has two children while every internal node for an atomic or truncation rule has a single child.
  \end{itemize}
\end{proposition}

\begin{figure}[t]
  \centering
  \includegraphics[scale=0.46]{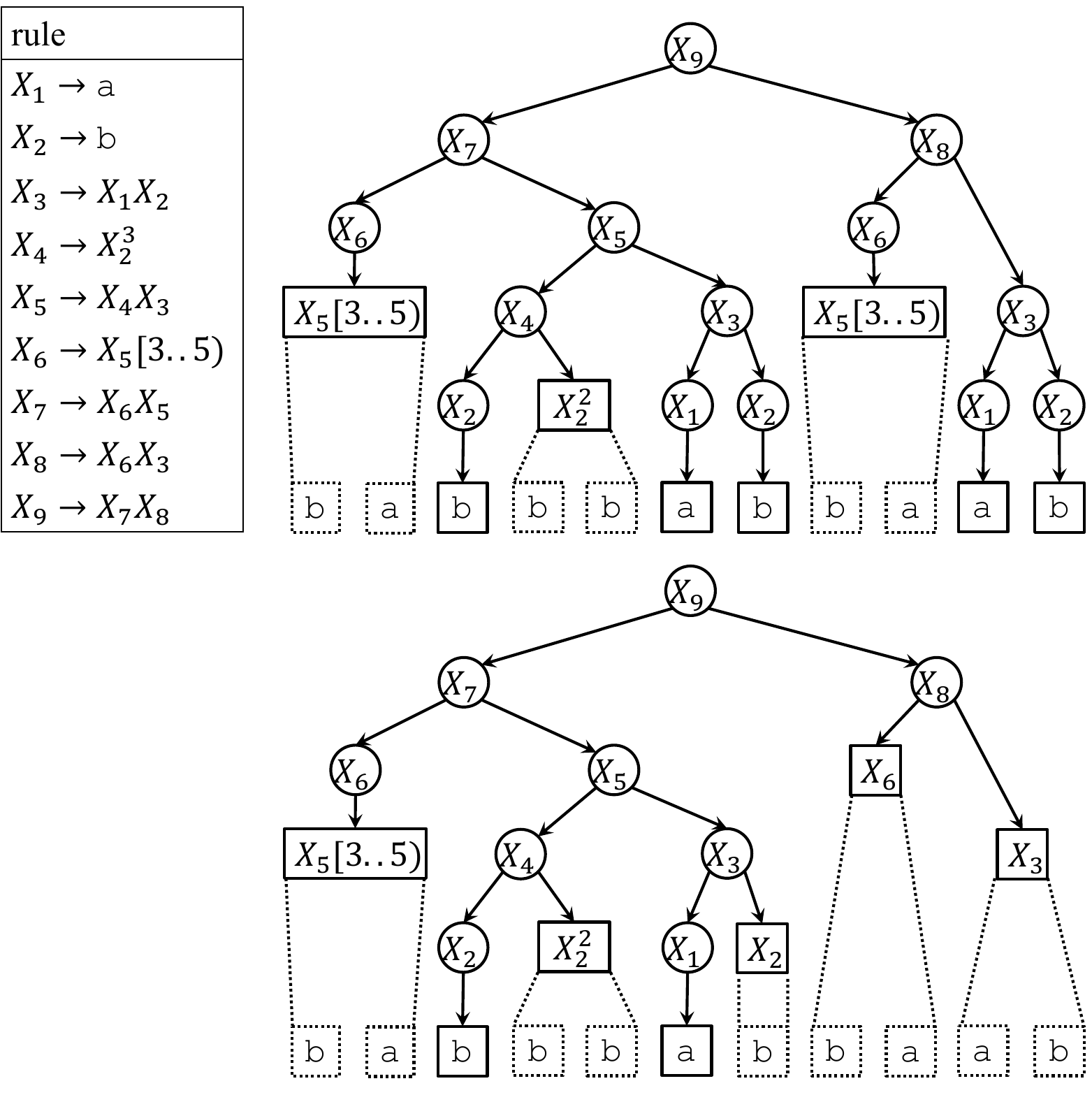}
  \caption{
    An illustration of the binary parse tree (right above) and the grammar tree (right below) for the internal collage system having 9 rules shown left.
    Three internal nodes with label $X_2$, $X_3$ and $X_6$, turn into leaves in the grammar tree because they are not the leftmost nodes with their own label.
    Observe that \cref{prop:ics} holds, i.e., the number of internal nodes is the size $m = 9$ of the collage system, and the number of leaves is $m - \numtr - \sigma + 1 = 9 - 1 - 2 + 1 = 7$, where $\sigma = 2$ is the alphabet size and $\numtr = 1$ is the number of truncation rules.
  }
  \label{fig:grammar_tree}
\end{figure}

\subsection{MAX-SAT}
Given a set of clauses, 
the satisfiability (SAT) problem asks for an assignment of variables that satisfies all the clauses.
An extension to SAT is MAX-SAT, where we are given \emph{hard clauses} and \emph{soft clauses},
and maximize the number of satisfying soft clauses while satisfying all the hard clauses.

\section{Collage Systems to Internal Collage Systems}\label{sec:cs2ics}
In this section, we show that there is an algorithm to convert a collage system of size $m$ for a string $T$ to an internal collage system of size $O(m)$, and prove that $\cint(T) = \Theta(\cgen(T))$.

\begin{figure}[t]
  \centering
  \includegraphics[scale=0.46]{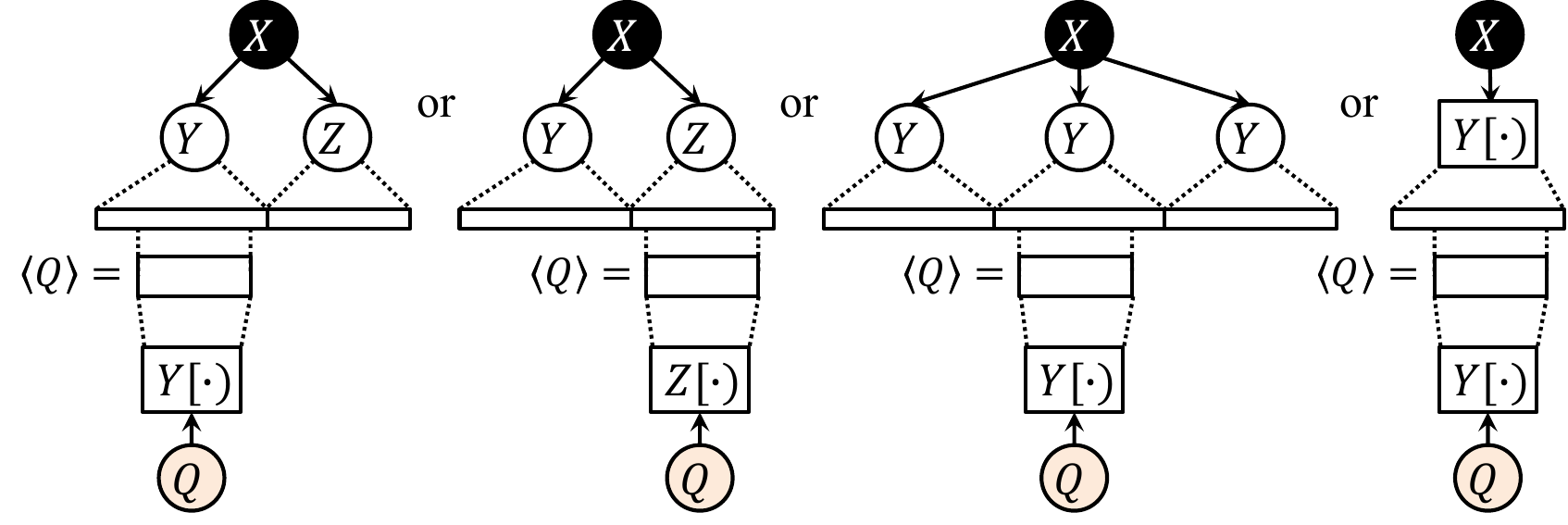}
  \caption{
    Illustration for Case~1.
    The upper parts show where the truncated substring $\expand{Q}$ exists in $\expand{X}$ and the lower parts show how the converted collage system represents it without using $X$.
    The intervals in the truncation rules are abbreviated and shown as ``$[\cdot)$''.
    The reference of a nonterminal $Q$ is changed from an unreachable nonterminal $X$ to a smaller nonterminal $Y$ or $Z$.
    If the new reference is unreachable, it will be processed later.
  }
  \label{fig:case1}
\end{figure}

\begin{figure}[t]
  \centering
  \includegraphics[scale=0.46]{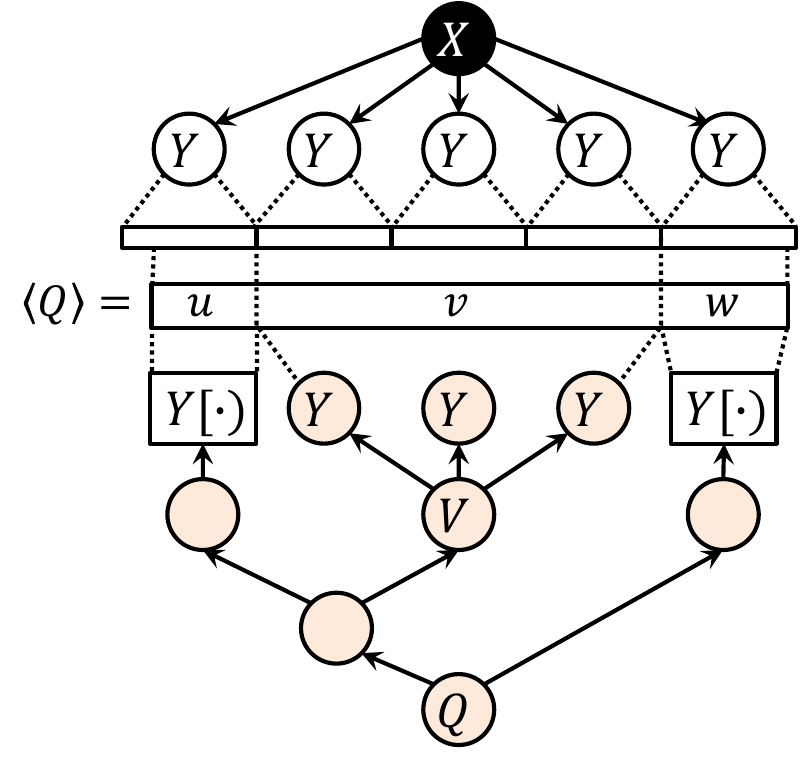}
  \caption{
    Illustration for Case~2.
    At most four new nonterminals including $V$ are enough to represent $\expand{Q}$.
    Note that new rules that truncate $Y$ are introduced, but $Y$ is guaranteed to be reachable via $Q$.
  }
  \label{fig:case2}
\end{figure}

\begin{figure}[t]
  \centering
  \includegraphics[scale=0.46]{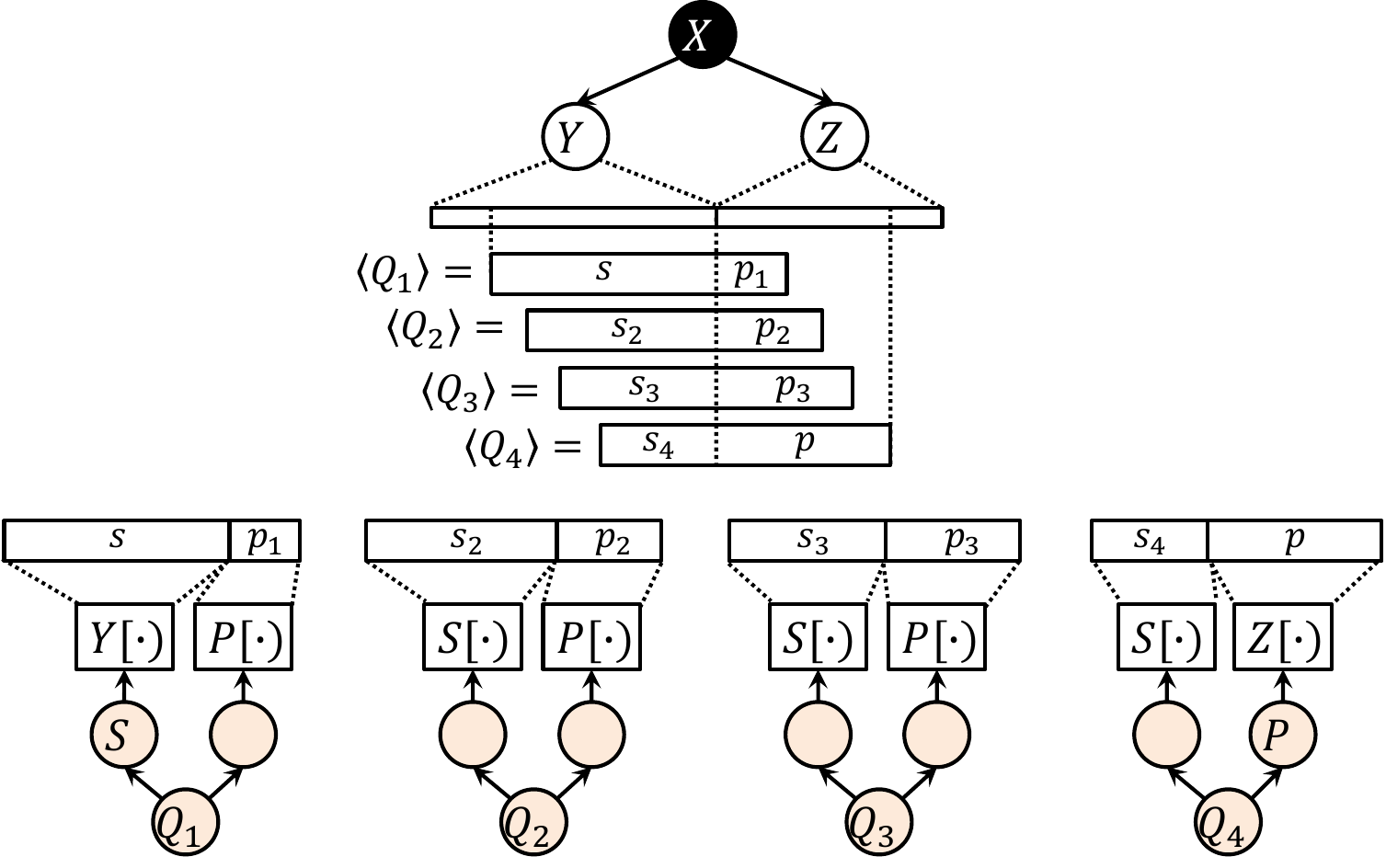}
  \caption{
    Illustration for Case~3, where the set $A$ (described in the proof of \cref{theo:cs2ics}) is assumed to contain four truncation rules with lefthand sides $Q_1, Q_2, Q_3$ and $Q_4$.
    For each $Q_i~(1 \le i \le 4)$, we introduce at most two new nonterminals, which are guaranteed to be reachable via $Q_i$.
    All such nonterminals other than $S$ and $P$ refer to a reachable nonterminal $S$ or $P$.
    The new truncation rule $S \rightarrow Y[\cdot)$ (resp. $P \rightarrow Z[\cdot)$) will be processed later if $Y$ (resp. $Z$) is unreachable.
  }
  \label{fig:case3}
\end{figure}

\begin{theorem}\label{theo:cs2ics}
  There is an algorithm to convert a collage system of size $m$ for a string $T$ to an internal collage system of size $O(m)$ in $O(m^2)$ time.
\end{theorem}
\begin{proof}
  We call a nonterminal \emph{reachable} if it appears as a node label of the grammar tree, and otherwise \emph{unreachable}.
  Then, a collage system is internal if and only if every nonterminal is reachable.
  Our algorithm processes unreachable nonterminals $X$ in a top-down manner and modifies the collage system to represent $T$ without using $X$.
  When we get into an unreachable nonterminal $X$, all nonterminals above have been processed to be reachable, and therefore, $X$ is referred to only by truncation rules because referring to $X_i$ in a rule other than truncation rules implies that $X$ is reachable.
  Hence, removing $X$ from the collage system is equivalent to removing the truncation rules that refer to $X$.
  The process of removing these truncation rules depends on the following cases:
  \begin{itemize}
    \item Case~0: When the truncation rule can be replaced with an atomic rule.
    \item Case~1: When the truncation can be made from a single smaller nonterminal.
    \item Case~2: When the truncation can be made from more than two repeating nonterminals.
    \item Case~3: When the truncation can be made from two nonterminals.
  \end{itemize}

  \textbf{Case~0:}
  Case~0 is the base case where $X \rightarrow a$ is an atomic rule.
  For any nonterminal $Q$ that truncates $X$, it means that $\expand{Q} = a$ and it is obvious that the truncation rule can be replaced with $Q \rightarrow a$.

  \textbf{Case~1:}
  Suppose that we want to remove a truncation rule $Q \rightarrow X[b..e)$.
  Case~1 deals with the following situations:
  \begin{itemize}
    \item When $X \rightarrow YZ$ and the interval $[b..e)$ is included in either one of $[1..|\expand{Y}|]$ or $[|\expand{Y}|+1..|\expand{Z}|]$.
    \item When $X \rightarrow Y^{r}$ and the interval $[b..e)$ is covered by a single $Y$, i.e., $[b..e)$ is included in $[(r'-1)|\expand{Y}| + 1.. r'|\expand{Y}|]$ for some $r' \in [1..r]$.
    \item When $X \rightarrow Y[b'..e')$.
  \end{itemize}
  In either case, it is easy to replace the truncation rule $Q \rightarrow X[b..e)$ with a rule that truncates a smaller nonterminal $Y$ or $Z$. If the smaller nonterminal is unreachable, the added truncation rule will be processed later when we get into the nonterminal.

  \textbf{Case~2:}
  Case~2 deals with a truncation rule $Q \rightarrow X[b..e)$ that truncates unreachable nonterminal $X$ with repetition rule $X \rightarrow Y^{r}$ spanning more than two $Y$'s.
  In this case, $\expand{Q}$ can be represented by $uvw$, where $u$ is a (potentially empty) suffix of $\expand{Y}$, $v = \expand{Y}^{r'}$ with $r' \ge 1$ and $w$ is a (potentially empty) prefix of $\expand{Y}$.
  We add a new nonterminal $U$ (resp. $W$) with a rule that truncates $Y$ to represent $u$ (resp. $w$) if $u$ (resp. $w$) is not the empty string.
  We also add a new nonterminal $V$ with repetition rule $V \rightarrow Y^{r'}$ if $r' > 1$.\footnote{If $r' = 2$, the added rule is categorized to concatenation rules in our definition.}
  Finally, we can represent $Q$ by concatenating at most three nonterminals for $u$, $v$ and $w$ using at most two concatenation rules.
  The process increases the size of the collage system by at most four because we create at most five rules and remove the original truncation rule $Q \rightarrow X[b..e)$.
  Note that $Y$ is reachable because $Q$ is reachable by inductive hypothesis and there is a path from $Q$ to $Y$ without truncation.
  Hence, added rules that truncate $Y$ do not need to be processed later.

  \textbf{Case~3:}
  Case~3 deals with rules that truncate an unreachable nonterminal $X$ with concatenation or repetition rule spanning exactly two nonterminals.
  We focus on the case with concatenation rule $X \rightarrow YZ$ as the case with repetition rule can be processed similarly.
  Let $A$ be the set of all rules that truncate $X$ and their truncated intervals are not included in $[1..|\expand{Y}|]$ nor $[|\expand{Y}|+1..|\expand{X}|]$.
  Unlike other cases, we process all truncation rules in $A$ at once.

  Let $s$ (resp. $p$) be the length of the longest suffix (resp. prefix) truncated from $Y$ (resp. $Z$) by a truncation rule in $A$.
  We create a new nonterminal $S$ that truncates the suffix of length $s$ of $Y$ and $P$ that truncates the prefix of length $p$ of $Z$.
  Then, any truncation made by a rule in $A$ can be replaced with concatenation of a nonterminal that truncates $S$ (or $S$ itself) and a nonterminal that truncates $P$ (or $P$ itself).
  This increases the size of the collage system by at most two per truncation rule in $A$.
  Note that any nonterminal, including $S$ and $P$, created during the process is reachable through the nonterminals in the lefthand sides of $A$, which are reachable by inductive hypothesis.
  Also, $S$ and $P$ are the only nonterminals that are created in the process and may truncate unreachable nonterminals, which will be processed later.

  \textbf{Correctness of conversion:}
  Each modification to the collage system does not change the string it represents.
  Since all nonterminals added during the process are reachable, the collage system becomes internal when we go through the process of removing unreachable nonterminals in the original collage system.

  \textbf{Size analysis:}
  The number of nonterminals in the collage system is increased by the process of removing a rule that truncates unreachable nonterminals by four in Case~2 and by two in Case~3.
  We have to take care of additional rules $S$ and $P$ created in Case~3 that may truncate unreachable nonterminals.
  The number of such additional truncation rules is bounded by $2(m-\numtr)$ in total, where $\numtr$ is the number of truncation rules in the original collage system.
  Thus, the size of the resulting internal collage system is bounded by $m + 4 (2(m-\numtr) + \numtr) = 9m - 4\numtr \le 9m$.

  \textbf{Time complexity:}
  In every case, we can process a truncation rule in $O(1)$ time each.
  While the truncation rule is completely removed in Case~0, Case~2 and Case~3,
  what we do in Case~1 is just switching its reference nonterminal to smaller nonterminals in $O(1)$ time.
  This switching cost may add up to $O(m)$ until we finally remove it in Case~0, Case~2 or Case~3.
  Since we remove $O(m)$ truncation rules, the algorithm runs in $O(m^2)$ time.
\end{proof}

It follows from \Cref{theo:cs2ics} that $\cint(T) \le 9\cgen(T)$.
Together with the fact that $\cgen(T) \le \cint(T)$, we get the following theorem:
\begin{theorem}\label{theo:ics_size}
  For any string $T$, it holds that $\cint(T) = \Theta(\cgen(T))$.
\end{theorem}

\section{MAX-SAT formulation to compute $\cint(T)$}\label{sec:maxsat}
In this section, we give a MAX-SAT formulation to compute $\cint(T)$ for a given string $T$.
The formulation is based on the \emph{ICS-factorization} $T = F_1 \cdots F_h$ defined for an internal collage system such that $F_i$ is the string derived from the $i$-th leaf of its grammar tree.
Recall that $h = m - \numtr - \sigma + 1$ holds by \Cref{prop:ics}, where $m$ is the size and $\numtr$ is the number of truncation rules of the internal collage system.
Thus, minimizing the size $m = h + \numtr + \sigma - 1$ is reduced to finding a valid ICS-factorization with the fewest factors and truncation rules.

\begin{example}
  The ICS-factorization of $T = \idtt{babbbabbaab}$ for the internal collage system of \cref{fig:grammar_tree} has seven factors with $F_1 = \idtt{ba}$, $F_2 = \idtt{b}$, $F_3 = \idtt{bb}$, $F_4 = \idtt{a}$, $F_5 = \idtt{b}$, $F_6 = \idtt{ba}$ and $F_7 = \idtt{ab}$.
\end{example}

The following Lemma gives a necessary and sufficient condition for a valid ICS-factorization, which extends Lemma~4.1 of~\cite{2024KawamotoIKB_HardnOfSmallRlslpAnd_DCC} for RLSLPs to incorporate truncation rules.
\begin{lemma}\label{lemma:ics_factorization_and_reference_structure}
  Let $T = F_1\cdots F_h$ be a factorization of a string $T$ and let $s_i = 1 + \sum_{k=1}^{i-1}|F_k|$ for $1 \le i \le h+1$.
  The factorization is a valid ICS-factorization if and only if each factor $F_k$ longer than $1$ is categorized into either one of the types (A), (B) and (C) so that the following conditions hold:
  \begin{enumerate}[(i)]
    \item Any factor $F_k$ longer than $1$ satisfies the following depending on its type:
    \begin{enumerate}[(A)]
      \item $F_k = T\hbrk{s_{i_k}..s_{j_k+1}}$ for some integers $i_k \le j_k < k$.
      \item $F_k = (T\hbrk{s_{i_k}..s_{j_k+1}})^{r-1}$ for some integers $i_k \leq j_k = k - 1$ and $r \geq 3$.
			\item $F_k$ is a substring of $T\hbrk{s_{i_k}..s_{j_k+1}}$ for some integers $i_k \le j_k$ with $k < i_k$ or $k > j_k$.
    \end{enumerate}
    \item The set $I = I_1 \cup I_2 \cup I_3 \cup I_4$ of intervals is compatible with tree structures, where
    \begin{itemize}
      \item $I_1 = \{\hbrk{s_{i_k}..s_{j_k+1}} \mid F_k \mbox{ is type-A}\}$,
      \item $I_2 = \{\hbrk{s_{i_k}..s_{j_k+1}} \mid F_k \mbox{ is type-B}\}$,
      \item $I_3 = \{\hbrk{s_{i_k}..s_{k+1}} \mid F_k \mbox{ is type-B}\}$, and
      \item $I_4 = \{\hbrk{s_{i_k}..s_{j_{k+1}}} \mid F_k \mbox{ is type-C}\}$.
    \end{itemize}
    \item The interval for a Type-B factor is not in $I$.
    \item Each factor $F_k$ can be given an integer $D_k$ such that:
    \begin{itemize}
      \item $D_k = 0$ if $|F_k| = 1$.
      \item $D_k > \max \{D_{i_k}, \dots, D_{j_k}\}$ if $|F_k| > 1$.
    \end{itemize}
  \end{enumerate}
\end{lemma}
\begin{proof}
  \textbf{$(\Rightarrow)$:}
  Suppose that $T = F_1 \cdots F_h$ is the ICS-factorization of an internal collage system $G$.
  Let $L_i$ denote the $i$-th leaf of the grammar tree of $G$.
  We categorize a factor $F_k$ longer than $1$ and set integers $i_k$ and $j_k$ as follows.
  \begin{itemize}
    \item If $L_k$ is labeled with a nonterminal $X$, then $F_k$ is type-A.
          Since $L_k$ is not the leftmost node with label $X$,
          we can set integers $i_k \le j_k < k$ so that $L_{i_k}, L_{i_k + 1}, \dots, L_{j_k}$ are the leaves under the internal node with label $X$ to satisfy condition (i).
          The interval $\hbrk{s_{i_k}..s_{j_k+1}}$ in $I_1$ corresponds to the internal node $X$.
    \item If $L_k$ is the right child of a nonterminal $X$ with a repetition rule $X \rightarrow Y^{r}$, then $F_k$ is type-B.
          We can set integer $i_k$ so that $L_{i_k}, L_{i_k + 1}, \dots, L_{k-1}$ are the leaves under the left child of $X$ to satisfy condition (i).
          The interval $\hbrk{s_{i_k}..s_{j_k+1}}$ in $I_2$ corresponds to the left child of $X$, while the interval $\hbrk{s_{i_k}..s_{k+1}}$ in $I_3$ corresponds to the internal node $X$.
    \item If $L_k$ is the child of a nonterminal $X$ with a truncation rule $X \rightarrow Y\hbrk{b..e}$, then $F_k$ is type-C.
          We can set integers $i_k$ and $j_k$ so that $L_{i_k}, L_{i_k + 1}, \dots, L_{j_k}$ are the leaves under the internal node with label $Y$ to satisfy condition (i).
          The interval $\hbrk{s_{i_k}..s_{j_{k+1}}}$ in $I_4$ corresponds to the internal node $Y$.
  \end{itemize}

  As seen above, every interval in $I$ corresponds to a node of the grammar tree, and hence, it is compatible with tree structures.
  Also, the interval of a type-B factor corresponds to a leaf node labeled with $Y^{r-1}$, which implies that it is not in $I$.
  Finally, we prove condition (iv).
  For any factor $F_k$ longer than $1$, we have set $i_k$ and $j_k$ so that $F_k$ is derived from a nonterminal that refers to smaller nonterminals deriving $F_{i_k} \cdots F_{j_k}$.
  Since $G$ does not have a reference loop of nonterminals, there is no reference loop of factors, too.
  Hence, we can assign integer $D_k$ to $F_k$ so that $D_k$ is larger than integers $\{D_{i_k}, \dots, D_{j_k}\}$ assigned to factors it refers to.
  Assigning $0$ to a factor of length $1$, we see that condition (iv) holds.

  \textbf{$(\Leftarrow)$:}
  Given that every factor $F_k$ longer than $1$ is categorized into types to satisfy conditions (i), (ii), (iii) and (iv).
  We can define a reference structure of factors according to condition (i).
  Referring to a type-B factor alone is avoided by condition (iii), and it is guaranteed by condition (iv) that there is no reference loop of factors.
  By condition (ii), $I$ can identify the intervals corresponding to nodes of a tree,
  which are sufficient to define a reference structure of nonterminals.
  Although not all nodes are identified, we can fill missing internal nodes by adding concatenation rules to build an internal collage system for the ICS-factorization.
\end{proof}

We modify the MAX-SAT formulation in~\cite{2024KawamotoIKB_HardnOfSmallRlslpAnd_DCC} for RLSLPs to encode valid ICS-factorizations of smallest internal collage systems based on \Cref{lemma:ics_factorization_and_reference_structure}.
We say that an ICS-factor longer than $1$ \emph{refers to} the substring $T\hbrk{s_{i_k}..s_{j_{k+1}}}$ defined in \Cref{lemma:ics_factorization_and_reference_structure} for each type.
Let us call an integer $D_k$ satisfying condition (iii) of \Cref{lemma:ics_factorization_and_reference_structure} the \emph{reference depth} of factor $F_k$ 
or the characters covered by $F_k$.
Note that a type-C factor may refer to a string to the right while type-A and type-B factors always refer to a substring to the left, which adds some complication to our formulation compared to the previous one for RLSLPs.
We will encode reference depths to ensure that there is no reference loop.

In the next paragraph, we will define several Boolean variables with multiple parameters in subscripts.
For any variable $x$ with some missing subscripts specified by ``$\circ$'', 
let $\sub{x}$ denote the set of (possibly tuples of) feasible subscripts to fill the missing part(s).
For example, $\sub{\litrefa{\circ}{\circ}{\circ}} = \{ (i', i, \ell) \mid i \in [1..n-1], \ell \in [2..n-i+1], i' \in [1..i-\ell], T[i'..i'+\ell) = T[i..i+\ell)\}$, and for some fixed $i$ and $\ell$, $\sub{\litrefa{\circ}{i}{\ell}} = \{ i' \mid (i', i, \ell) \in \sub{\litrefa{\circ}{\circ}{\circ}}\}$.
We also use ``$\bullet$'' to represent anonymous (arbitrary) subscripts, which are filtered out when used in $\sub{x}$.
For example, $\sub{\litrefa{\bullet}{\circ}{\circ}} = \{ (i, \ell) \mid (i', i, \ell) \in \sub{\litrefa{\circ}{\circ}{\circ}} \}$.

For a given text $T[1..n]$ of length $n$, we define Boolean variables as follows to encode \Cref{lemma:ics_factorization_and_reference_structure}.
\begin{itemize}
  \item $\litf{i}{\ell}$ for $i \in [1..n]$ and $\ell \in [1..n+1-i]$: $\litf{i}{\ell} = 1$ iff $T[i..i+\ell)$ is an ICS-factor.
  \item $\litp{i}$ for $i \in [1..n+1]$: For $i \neq n+1$, $\litp{i} = 1$ iff $i$ is the starting position of an ICS-factor.
        $\litp{n+1}$ exists for technical reasons. We set $\litp{1} = \litp{n+1} = 1$.
  \item $\litrefa{i'}{i}{\ell}$ for $i \in [1..n-1]$, $\ell \in [2..n - i + 1]$ and $i' \in [1..i-\ell]$ s.t. $T[i'..i'+\ell) = T[i..i+\ell)$:
        $\litrefa{i'}{i}{\ell} = 1$ iff $T[i..i+\ell)$ is a type-A factor that refers to $T[i'..i'+\ell)$.
  \item $\litrefb{i'}{i}{\ell}$ for $i \in [1..n-1]$, $\ell \in [2..n-i+1]$ and $i' \in [i-\ell+1..i-1]$ s.t. $T[i'..i'+\ell) = T[i..i+\ell)$ and $i - i'$ divides $\ell$:
        $\litrefb{i'}{i}{\ell} = 1$ iff $T[i..i+\ell)$ is a type-B factor that refers to $T[i'..i)$ to repeat it.
  \item $\litrefc{i'}{\ell'}{i}{\ell}$ for $i, i' \in [1..n-2]$, $\ell \in [2..n-i+1]$, $\ell' \in [2..n-i'+1]$ and $\hbrk{i..i+\ell} \cap \hbrk{i'..i'+\ell'} = \emptyset$ s.t. $T[i..i+\ell)$ is a substring of $T[i'..i'+\ell')$:
        $\litrefc{i'}{\ell'}{i}{\ell} = 1$ iff $T[i..i+\ell)$ is a type-C factor that refers to $T[i'..i'+\ell')$ to truncate it.
  \item $\litdref{i'}{\ell'}{i}$ for 
        $(i', i, \ell) \in \sub{\litrefa{\circ}{\circ}{\circ}}$ with $\ell' = \ell$, 
        $(i', i) \in \sub{\litrefb{\circ}{\circ}{\bullet}}$ with $\ell' = i - i'$, or
        $(i', \ell', i) \in \sub{\litrefc{\circ}{\circ}{\circ}{\bullet}}$:
        $\litdref{i'}{\ell'}{i} = 1$ iff there is an ICS-factor that starts at $i$ and refers to $T\hbrk{i'..i'+\ell'}$.
  \item $\litq{i'}{\ell'}$ for $(i', \ell') \in \sub{\litdref{\circ}{\circ}{\bullet}}$ or 
        $(i', i, \ell) \in \sub{\litrefb{\circ}{\circ}{\circ}}$ with $\ell' = \ell + i - i'$: $\litq{i'}
        {\ell'} = 1$ iff $\hbrk{i'..i'+\ell'}$ is in the set $I$ defined in \Cref{lemma:ics_factorization_and_reference_structure}.
  \item $\litdepth{i}{\ell}{d}$ for $i \in [1..n]$, $\ell \in [1..n+1-i]$ and $d \in [0..n]$: 
        $\litdepth{i}{\ell}{d} = 1$ iff $d$ is less than or equal to the maximum reference depth of characters in $T\hbrk{i..i+\ell}$.
\end{itemize}

We next define constraints that the above variables must satisfy.

Since the factors and their starting positions must be consistent, we have:
\begin{equation}
    \forall (i, \ell) \in \sub{\litf{\circ}{\circ}}: \litf{i}{\ell} \Longleftrightarrow \litp{i} \land (\neg \litp{i+1}) \land \cdots \land (\neg \litp{i+\ell-1}) \land \litp{i+\ell}.
    \label{eqn:litf}
\end{equation}

If $T[i..i+\ell)$ with $\ell > 1$ cannot be a candidate of type-A, type-B nor type-C factor,
$\litf{i}{\ell}$ must be false, i.e.,
\begin{equation}
	\forall (i, \ell) \in \sub{\litf{\circ}{\circ}} \setminus (\sub{\litrefa{\bullet}{\circ}{\circ}} \cup \sub{\litrefb{\bullet}{\circ}{\circ}} \cup \sub{\litrefc{\bullet}{\bullet}{\circ}{\circ}}) \mbox{ with } l > 1 : \lnot \litf{i}{l}.
  \label{eqn:neg_litf}
\end{equation}

If $T[i..i+\ell)$ is a type-A factor, there exists $i' \in \litrefa{\circ}{i}{\ell}$ s.t. $\litrefa{i'}{i}{\ell}$ is true.
If $T[i..i+\ell)$ is a type-B factor, there exists $i' \in \litrefb{\circ}{i}{\ell}$ s.t. $\litrefb{i'}{i}{\ell}$ is true.
If $T[i..i+\ell)$ is a type-C factor, there exists $(i', \ell') \in \litrefc{\circ}{\circ}{i}{\ell}$ s.t. $\litrefc{i'}{\ell'}{i}{\ell}$ is true.
Conversely, if one of the variables of the form $\litrefa{i'}{i}{\ell}$, $\litrefb{i'}{i}{\ell}$ or $\litrefc{i'}{\ell'}{i}{\ell}$ is true, $T[i..i+\ell)$ must be a factor. This can be encoded as
\begin{align}
  &
  \forall (i, \ell) \in \sub{\litrefa{\bullet}{\circ}{\circ}} \cup \sub{\litrefb{\bullet}{\circ}{\circ}} \cup \sub{\litrefc{\bullet}{\bullet}{\circ}{\circ}}: 
  \nonumber \\
  &
  \litf{i}{\ell} \Longleftrightarrow
  \left( \bigvee_{i' \in \sub{\litrefa{\circ}{i}{\ell}}} \hspace{-1.5em} \litrefa{i'}{i}{\ell} \right) \vee 
  \left( \bigvee_{i' \in \sub{\litrefb{\circ}{i}{\ell}}} \hspace{-1.5em} \litrefb{i'}{i}{\ell} \right) \vee
  \left( \bigvee_{(i', \ell') \in \sub{\litrefc{\circ}{\circ}{i}{\ell}}} \hspace{-2.3em} \litrefc{i'}{\ell'}{i}{\ell} \right).
  \label{eqn:litf_litref}
\end{align}

For a fixed substring $T[i..i+\ell)$, at most one variable of the form $\litrefa{i'}{i}{\ell}$, $\litrefb{i'}{i}{\ell}$ or $\litrefc{i'}{\ell'}{i}{\ell}$ is allowed to indicate that $T[i..i+\ell)$ is a factor. Thus, we require
\begin{align}
  &
  \forall (i, \ell) \in \sub{\litrefa{\bullet}{\circ}{\circ}} \cup \sub{\litrefb{\bullet}{\circ}{\circ}} \cup \sub{\litrefc{\bullet}{\bullet}{\circ}{\circ}}: \nonumber \\
  & 
  \sum_{i' \in \sub{\litrefa{\circ}{i}{\ell}}} \hspace{-1.5em} \litrefa{i'}{i}{\ell} +
  \sum_{i' \in \sub{\litrefb{\circ}{i}{\ell}}} \hspace{-1.5em} \litrefb{i'}{i}{\ell} +
  \sum_{(i', \ell') \in \sub{\litrefc{\circ}{\circ}{i}{\ell}}} \hspace{-2.3em} \litrefc{i'}{\ell'}{i}{\ell} \le 1.
  \label{eqn:litref_sum}
\end{align}
Constraint~\ref{eqn:litref_sum} for a fixed $(i, \ell)$ can be encoded in size of $\left| \sub{\litrefa{\circ}{i}{\ell}} \cup \sub{\litrefb{\circ}{i}{\ell}} \cup \sub{\litrefc{\circ}{\circ}{i}{\ell}} \right| = O(n^2)$ using an efficient encoding for this kind of ``at-most'' constraints~\cite{2005Sinz_TowarOptimCnfEncodOf_CP}.

If $\litdref{i'}{\ell'}{i}$ is true, then there is a factor that refers to $T[i'..i'+\ell')$ and vice versa.
\begin{align}
  &
  \forall (i', \ell', i) \in \sub{\litdref{\circ}{\circ}{\circ}}:
  \nonumber \\
  &
  \litdref{i'}{\ell'}{i} \Longleftrightarrow
  \litrefa{i'}{i}{\ell'} \vee
  \left( \bigvee_{\ell \in \sub{\litrefb{i'}{i}{\circ}}} \litrefb{i'}{i}{\ell} \right) \vee
  \left( \bigvee_{\ell \in \sub{\litrefc{i'}{\ell'}{i}{\circ}}} \litrefc{i'}{\ell'}{i}{\ell} \right).
  \label{eqn:litdref_litref}
\end{align}

We let $\litq{i'}{\ell}$ summarize the information on whether
$T[i'..i'+\ell')$ corresponds to a node implied by
$\litrefa{\bullet}{\bullet}{\bullet}$, $\litrefb{\bullet}{\bullet}{\bullet}$ and $\litrefc{\bullet}{\bullet}{\bullet}{\bullet}$:
\begin{equation}
    \forall (i', \ell') \in \sub{\litq{\circ}{\circ}}:
    \litq{i'}{\ell'} \Longleftrightarrow
    \left( \bigvee_{i \in \sub{\litdref{i'}{\ell'}{\circ}}} \hspace{-1.9em} \litdref{i'}{\ell'}{i} \right) \vee
    \left( \bigvee_{\substack{(i, \ell) \in \sub{\litrefb{i'}{\circ}{\circ}}\\ \mbox{with } \ell' = \ell + i - i'}}
           \hspace{-1.9em} \litrefb{i'}{i}{\ell} \right).
    \label{eqn:litq_litref}
\end{equation}

If $\litq{i'}{\ell'}$ is true, then $i'$ must be the starting position of an ICS-factor.
\begin{eqnarray}
    \forall (i', \ell') \in \sub{\litq{\circ}{\circ}}:
    \litq{i'}{\ell'}
    \Longrightarrow
    \litp{i'}.
    \label{eqn:litq_litp}
\end{eqnarray}

If $\litq{i'}{\ell'}$ is true, then $T[i'..i'+\ell')$ must not be a type-B factor:
\begin{eqnarray}
    \forall (i', \ell') \in \sub{\litq{\circ}{\circ}}, \forall i'' \in \sub{\litrefb{\circ}{i'}{\ell'}}:
    \litq{i'}{\ell'}
    \Longrightarrow
    \lnot \litrefb{i''}{i'}{\ell'}.
    \label{eqn:litq_notlitrefb}
\end{eqnarray}

The set of intervals indicated by $\litq{\bullet}{\bullet}$ must be compatible with tree structures.
In other words, for any two substrings $T[i'_1..i'_1+\ell'_1)$ and $T[i'_2..i'_2+\ell'_2)$
with $i'_1 < i'_2 < i'_1 + \ell'_1 < i'_2 + \ell'_2$,
at most one of $T[i'_1..i'_1+\ell'_1)$ and $T[i'_2..i'_2+\ell'_2)$
can correspond to a node of the grammar tree.
Thus, we require that
\begin{eqnarray}
    && \forall (i'_1, \ell'_1), (i'_2, \ell'_2) \in \sub{\litq{\circ}{\circ}} \mbox{ s.t. } i'_1 < i'_2 < i'_1 + \ell'_1 < i'_2 + \ell'_2: 
    \neg \litq{i'_1}{\ell'_1} \lor \neg \litq{i'_2}{\ell'_2}.
    \label{eqn:litq_nest}
\end{eqnarray}

Since the reference depth of a character is at least 0, $\litdepth{i}{1}{0}$ is always true:
\begin{align}
	\forall i \in [1..n]: \litdepth{i}{1}{0} & = 1.
	\label{eqn:litdepth_0}
\end{align}

If the reference depth of a character is at least $d$, then it is also at least $d-1$:
\begin{align}
	\forall i \in [1..n], \forall d \in [1..n]: \litdepth{i}{1}{d} \Longrightarrow \litdepth{i}{1}{d-1}.
  \label{eqn:litdepth_less}
\end{align}

The reference depths of characters in the same factor are the same:
\begin{align}
	\forall i \in [2..n], \forall d \in [0..n]: \neg p_{i} \Longrightarrow \litdepth{i}{1}{d} = \litdepth{i-1}{1}{d}.
	\label{eqn:litdepth_factor}
\end{align}

The reference depth of a substring is the maximum of the reference depths of characters in the substring:
\begin{align}
	\forall i \in [1..n], \forall \ell \in [1..n+1-i], \forall d \in [0..n]: \litdepth{i}{\ell}{d}\Longleftrightarrow \bigvee_{j=i}^{i+l-1} \litdepth{j}{1}{d}.
	\label{eqn:litdepth_max}
\end{align}

If an ICS-factor starting at $i$ refers to substring $T[i'..i'+\ell)$,
the reference depth of $T[i]$ is larger than the reference depth of $T[i'..i'+\ell)$, 
which is encoded as
\begin{align}
	\forall (i',\ell',i) \in \sub{\litdref{\circ}{\circ}{\circ}}:
  \litdref{i'}{\ell'}{i} \Longrightarrow \forall d \in [1..n] (\lnot \litdepth{i}{1}{d} \Longrightarrow \lnot \litdepth{i'}{\ell'}{d-1}).
	\label{eqn:litdepth_dref}
\end{align}

From the above SAT formulation, we get the following theorem:
\begin{theorem}
  There is a MAX-SAT formulation of size $O(n^4)$
  to compute smallest internal collage systems for a string $T$ of length $n$.
\end{theorem}
\begin{proof}
  \textbf{Correctness of SAT formulation:}
  We prove the correctness of our SAT formulation presented in this section for ICS-factorizations.
  On the one hand, given an internal collage system,
  it is clear that a truth assignment to Boolean variables based on the definition will satisfy all constraints.
  On the other hand, given $T$ and a truth assignment satisfying the constraints,
  we can deduce a factorization to satisfy the conditions (i), (ii), (iii) and (iv) of \Cref{lemma:ics_factorization_and_reference_structure}.
  First, the truth assignments of $\litp{i}$ and Constraint~(\ref{eqn:litf}) 
  give a factorization of $T$ where we regard $T[i..i+\ell)$ as a factor if and only if $\litf{i}{\ell} = 1$.
  If $\litf{i}{\ell} = 1$ with $\ell > 1$,
  the type of factor $T[i..i+\ell)$ and its reference to satisfy condition (i) are specified by the truth assignments of $\litrefa{\bullet}{\bullet}{\bullet}$, $\litrefb{\bullet}{\bullet}{\bullet}$ and $\litrefc{\bullet}{\bullet}{\bullet}{\bullet}$,
  which are suitably controlled by Constraints~(\ref{eqn:neg_litf}),~(\ref{eqn:litf_litref}) and~(\ref{eqn:litref_sum}).
  By Constraint~(\ref{eqn:litdref_litref}),
  the information about the starting position of a factor and its reference is collected in variables $\litdref{\bullet}{\bullet}{\bullet}$.
  By Constraint~(\ref{eqn:litq_litref}), the truth assignments of $\litq{\bullet}{\bullet}$
  give the set $I$ of intervals of \Cref{lemma:ics_factorization_and_reference_structure}.
  The starting position of every interval of $I$ must be the starting position of an ICS-factor, which is enforced by Constraint~(\ref{eqn:litq_litp}).
  Condition (iii) is fulfilled by Constraint~(\ref{eqn:litq_notlitrefb}).
  Constraint~(\ref{eqn:litq_nest}) ensures that $I$ is compatible with tree structures to satisfy condition (ii).
  Finally, condition (iv) is ensured by Constraints~(\ref{eqn:litdepth_0}), (\ref{eqn:litdepth_less}), (\ref{eqn:litdepth_factor}), (\ref{eqn:litdepth_max}) and (\ref{eqn:litdepth_dref}) using $\litdref{\bullet}{\bullet}{\bullet}$ and $\litdepth{\bullet}{\bullet}{\bullet}$.

  \textbf{MAX-SAT formulation:}
  By \Cref{prop:ics}, the size $m$ of an internal collage system corresponding to an ICS-factorization is
  $m = h + \numtr + \sigma - 1$, where $h$ is the number of factors, $\numtr$ is the number of truncation rules, and $\sigma$ is the number of distinct characters in $T$.
  Since $\sigma$ is fixed, $m$ is minimized if $h + \numtr$ is minimized.
  This can be encoded in MAX-SAT by introducing the set of soft clauses consisting of $\lnot \litp{\bullet}$ and $\lnot \litrefc{\bullet}{\bullet}{\bullet}{\bullet}$ 
  because satisfying these soft clauses as many as possible leads to minimizing $h + \numtr$.

  \textbf{Size analysis:}
  In total, we have $O(n^4)$ Boolean variables dominated by $\litrefc{\bullet}{\bullet}{\bullet}{\bullet}$.
  The total size of the resulting MAX-SAT formulation is $O(n^4)$,
  dominated by Constraints~(\ref{eqn:litf_litref}), (\ref{eqn:litref_sum}), (\ref{eqn:litdref_litref}), (\ref{eqn:litq_litref}), (\ref{eqn:litq_nest}), (\ref{eqn:litdepth_max}) and (\ref{eqn:litdepth_dref}) where there are four free parameters each of which take $O(n)$ different values.
\end{proof}

\section{Conclusions and future work}
In this paper, we proposed an $O(m^2)$-time algorithm to convert a collage system of size $m$ for a string $T$ to an internal collage system of size $\le 9m$.
As a consequence, we obtained new bounds $\cint(T) = \Theta(\cgen(T))$ and $b(T) = O(\cgen(T))$ on the smallest size of collage systems.
We also give a MAX-SAT formulation to compute $\cint(T)$ for a given $T$.

Future work includes the following directions:
\begin{itemize}
  \item Improve the current (probably loose) upper bound $9m$ for the size of converted internal collage system or give a nontrivial lower bound.
  \item Improve $O(m^2)$ running time for the conversion algorithm.
  \item Implement our MAX-SAT formulation and evaluate how well it scales compared to other formulations for SLPs~\cite{2022BannaiGIKKN_ComputNpHardRepetMeasur_ESA} and RLSLPs~\cite{2024KawamotoIKB_HardnOfSmallRlslpAnd_DCC}.
  \item Consider a formulation to compute $\cgen(T)$.
        Here, a new technique would be required because the current approach, inherited from~\cite{2022BannaiGIKKN_ComputNpHardRepetMeasur_ESA,2024KawamotoIKB_HardnOfSmallRlslpAnd_DCC}, depends on the fact that $\cint(T)$ is measured through the ICS-factorization of $T$, but it does not hold for general collage systems that have ``external'' references.
  \item Prove or disprove that $z(T) / \cgen(T) \in O(1)$, which is left unknown in the relation between compressibility measures presented in~\cite[Fig. 1]{2023KociumakaNP_TowarDefinComprMeasurFor}.
\end{itemize}

\bibliography{refs}

\end{document}